%% file: example.tex
\documentclass{article}

\usepackage[final]{corl_2022} 
\usepackage{amsthm,amsmath,amsfonts, dsfont, hyperref,graphicx}
\usepackage[dvipsnames]{xcolor}
\input{defs.tex}

\title{Reducing Collision Risk in Multi-Agent Path Planning: Application to Air Traffic Management}
\author{
  Sarah H.Q. Li\\
  Department of Aeronautics and Astronautics Engineering\\
  University of Washington, Seattle, United States\\
  \texttt{sarahli@uw.edu} \\
  \AND
  Avi Mittal \\
  Department of Aeronautics and Astronautics Engineering\\
  University of Washington, Seattle, United States \\
  \texttt{avim@uw.edu} \\
   \And
   Pierre-Loïc Garoche \\
   École Nationale de l'Aviation Civile \\
   Université de Toulouse, Toulouse, France \\
   \texttt{Pierre-Loic.Garoche@enac.fr} \\
   \And 
    A{\c{c}}{\i}kme{\c{s}}e, Beh{\c{c}}et\\
   Department of Aeronautics and Astronautics Engineering \\
   University of Washington, Seattle, United States \\
   \texttt{behcet@uw.edu} \\
}
\begin{document}
\maketitle
\begin{abstract}
    To minimize collision risks in the multi-agent path planning problem with stochastic transition dynamics, we formulate a Markov decision process congestion game with a multi-linear congestion cost. Players within the game complete individual tasks while minimizing their own collision risks. We show that the set of Nash equilibria  coincides with the first order KKT points of a non-convex optimization problem. Our game is applied to a historical flight plan over France to reduce collision risks between commercial aircraft. 
\end{abstract}
\input{intro.tex}
\input{mdp_model}
\input{flight_model}

\input{atm_results}

\section{Conclusion}
We derived an $N$-player MDP congestion game in which players solve MDPs that are coupled to the opponents through collision risk. We showed that its Nash equilibria are the KKT points of a potential minimization problem, and applied our model to collision risk reduction for commercial aircraft under operational uncertainty. Future work includes analyzing effect on flight delays.  

\clearpage


\bibliography{example}  

\end{document}

%% file: defs.tex
\newcommand{\reals}{{\mathbb{R}}}

\newcommand{\argmin}{\mathop{\rm argmin}}

\newcommand{\mnorm}[1]{{\left\vert\kern-0.25ex\left\vert\kern-0.25ex\left\vert #1 
    \right\vert\kern-0.25ex\right\vert\kern-0.25ex\right\vert}}

\newcommand{\mc}{\mathcal}

\newtheorem{definition}{Definition} 

\newtheorem{lemma}{Lemma}

%% file: intro.tex
\section{Introduction}
In robotics, aeronautics and warehouse logistics~\cite{ragi2013uav,trinh2020multi}, the operational dynamics are often inherently uncertain: delayed package arrivals may alter a warehouse's internal logistics, and quad-copters may be blown off its intended path by strong gusts of wind. 
In large-scale autonomy frameworks such as urban air mobility and automated warehouses, vehicles also experience \emph{congestion}---individual autonomous vehicles crowding the shared operational space.  Congestion can severely reduce system performance and requires inter-vehicle coordination to resolve. 
In~\cite{li2022congestion}, a potential game solution is proposed.
Using a heuristic to estimate work floor congestion,~\cite{li2022congestion} showed that multiple robots can share a work space with reduced collision risks.

Under uncertain operation dynamics, collision risks will always exist. While~\cite{li2022congestion} proves that individual task completion is optimal with respect to the congestion heuristic, no guarantee can be made on individual vehicle's collision risks. This prevents the adaptation of these path coordination games in safety-critical applications such air traffic management, in which government regulations require strict safety guarantees. To mitigate the lack of safety guarantee,  we consider an extension of the potential game model introduced in~\cite{li2022congestion,calderone2017markov} that directly minimizes the exact collision risk. By doing so, we can produce more rigorous guarantees on the individual vehicle's collision risks.

\textbf{Contributions}. Under the MDP congestion game framework~\cite{li2022congestion,calderone2017markov}, we propose a congestion model that directly weighs an autonomous vehicle's desire for task completion against its willingness to risk collisions. 
We show that this congestion cost has a potential function, so that the optimal multi-vehicle trajectory is the solution of a non-convex optimization problem with a multi-linear objective. 
We develop an in-depth game model of air traffic management using historical flight plans from the French air space and show that the Frank-Wolfe gradient descent method can find locally optimal solutions where individual flight's collision risk drops both in occurrence and in likelihood. 

%% file: mdp_model.tex
\section{Markov decision process congestion game}
\textbf{Notation}. We use $\reals(\reals_+)$ to denote the real (non-negative real) numbers, $[N]$ to denote $\{1,\ldots N\}$, $\mc{T}$ to denote $\{0,\ldots,T\}$, and $\Delta_N = \{y \in \reals_+^N \ | \ \sum_{i} y_i = 1\}$ to denote the simplex in $\reals^N$. 
\subsection{Individual Markov decision process (MDP)}
The \emph{finite-horizon MDP} for player $i$ is given by $([S], [A], \mc{T}, P^i, C^i, p^i_0)$, where $[S]$ is the \textbf{finite set of states}, $[A] $ is the \textbf{finite set of actions}, $\mc{T}$ is the \textbf{finite time horizon}, $C^i \in \reals^{(T+1)SA}$ are the \textbf{state-action costs}, $P^i \in \reals_+^{TSSA}$ is the \textbf{transition dynamics}, and $p^i_0 \in \Delta_S$ is player's \textbf{initial state probability distribution}. Assume that each action $a \in [A]$ is admissible from each state $s \in [S]$. 

At time $t \in \mc{T}$ and state $s \in [S]$, player $i$ selects an action $a \in [A]$ and incurs a cost $C_{tsa} \in \reals$. At $t+1$, the player transitions to state $s'$ with probability $P_{ts'sa}\geq 0$. This is repeated for $t\in\mc{T}$. At $t=0$, player $i$'s probability of being in state $s$ is given by $p_{0s}$.

Player $i$'s \textbf{state-action distribution} is $x^i \in \reals_+^{(T+1)SA}$, where $x^i_{tsa}$ is player's joint probability of taking action $a$ at $(t,s) \in \mc{T}\times [S]$. The set of feasible MDP state-action distributions is given by 
\begin{equation}
\textstyle \mc{X}(P^i, p^i_0) = \big\{ z \in \reals_+^{(T+1)SA} \ | \ \sum_{a}z_{0sa} = p^i_{0s}, 
     \sum_{a}z_{(t+1)sa} = \sum_{a, s'}P^i_{tss'a}z_{ts'a}, \forall (t,s) \in \mc{T}\times[S]\big\}.
\end{equation}


Player $i$'s \textbf{Q-value function} $Q^i\in\reals^{(T+1)SA}$ is the expected incurred cost within the MDP~\cite[Chp.4.2.1]{puterman2014markov}. When player $i$ is at state $s$ and time $t$, $Q^i_{tsa}$ is the expected total cost player incurs from state $s$ and time $t$ if it first takes action $a$ and plays optimally thereafter.
\begin{equation}\label{eqn:q_value}
   \textstyle  Q^i_{Tsa} := C^i_{Tsa}, \ 
  \textstyle   Q^i_{(t-1)sa}:= C^i_{(t-1)sa} + \sum_{s'} P^i_{ts'sa}\underset{a'}{\min}\, Q^i_{t,s'a'}, \ 
\forall \  (t, s, a) \in [T]\times[S]\times[A]. 
\end{equation}
\subsection{Multi-player MDPs under collision risk-based congestion}
Inspired by autonomous vehicles sharing an operation space, we consider the scenario in which $N$ players each solve the MDP $([S],[A],\mc{T}, P^i, \ell^i, p^i_{0})$ for $i \in [N]$. Distinct from individual MDPs, the MDP costs $\ell^i: \reals^{N\times (T+1)SA} \mapsto \reals^{(T+1)SA}$ depend on all players' state-action distributions $(x^1, \ldots, x^N)$. We denote this joint state-action distribution as $x = (x^1, \ldots, x^N)$ and the resulting cost as $\ell^i(x)$. The players jointly solve an \textbf{MDP congestion game} under costs $(\ell^1,\ldots,\ell^N)$. 

\textbf{Probabilistic collision risks}. A player's collision risk at $(t,s)$ and $(t,s,a)$ are the probabilities that at least one other player is in the same state and state-action, respectively. 
\begin{lemma}
 Under $x=(x^1,\ldots, x^N)$, player $i$'s probability of encountering at least one other player in $s$ and $(s,a)$ at time $t$ are respectively denoted by $D^i_{ts}(x)$ and $G^i_{tsa}(x)$ and computed as
 \begin{equation}
   \textstyle   D^i_{ts}(x)  = 1 - \prod_{j\neq i}(1 - \sum_{a'} x^j_{tsa'}), \ G^i_{tsa}(x) = 1  - \prod_{j\neq i}(1 -  x^j_{tsa}) \ \forall \ i, t, s, a \in [N]\times [T]\times [S]\times [A].\label{eqn:DG_prob}
 \end{equation}
\end{lemma}
\begin{proof}
 The probability of player $j$ taking state-action $(s, a)$ at time $t$ is $x^j_{tsa}$. The probability that player $j$ does \textit{not} take state-action $(s, a)$ at time $t$ is $1 - x^j_{tsa}$. The probability that \textit{none} of the players $j \neq i$ take state-action $(s, a)$ at time $t$ is $\textstyle \prod_{j \neq i} (1 - x^j_{tsa})$. The probability of \textit{at least one} other player $j \neq i$ taking state-action $(s, a)$ at time $t$ is given by $G^i_{tsa}(x)$ in~\eqref{eqn:DG_prob}. To derive $\textstyle D^i_{ts}(x)$~\eqref{eqn:DG_prob}, we apply similar arguments to the probability of player $j$ being in state $s$ at time $t$, given by $\textstyle \sum_a{x^j_{tsa}}$.
\end{proof}
As shown in Section~\ref{sec:flight_model}, $D^i$ and $G^i$ are flight separation constraints in air traffic management.

\textbf{Collision risk-based congestion}. We augment players' individual costs $C^i$ with $D^i(x)$ and $G^i(x)$. \begin{equation}\label{eqn:congestion_cost}
    \ell^i_{tsa}(x) = C^i_{tsa} + k\big(D^i_{ts}(x) + G^i_{tsa}(x)\big), \ \forall (t,s,a) \in \mc{T}\times[S]\times[A],
\end{equation}
where $k \in \reals_+$ is a user-defined parameter that signifies the players' willingness to risk collisions. Players are collision ignorant at $k=0$, and collision-averse at $k\rightarrow\infty$. Unique from~\cite{calderone2017markov,li2022congestion}, $\ell^i$~\eqref{eqn:congestion_cost} is independent of $x^i$; when player $i$'s opponents fix their strategies, player $i$ solves a standard MDP. 

When all players simultaneously achieve the minimum $Q^i(x)$, $x$ is a Nash equilibrium. 
\begin{definition}[Nash equilibrium]\label{def:NE}\cite{li2022congestion}
The state-action distribution $x$ is a Nash equilibrium if every player exclusively takes actions that minimize their $Q$-value function, $Q^i(x)$~\eqref{eqn:q_value}. 
\begin{equation}\label{eqn:nash}
   \textstyle  x^i_{tsa}> 0 \Rightarrow a \in \argmin\{Q^i_{tsa'}(x)\ | \ a' \in [A]\}, \ \forall (i, t, s, a) \in [N] \times [T]\times [S]\times [A]. 
\end{equation}
\end{definition}
We consider solving for the Nash equilibrium using the potential game formulation given by
\begin{equation}\label{eqn:potential_game}
\begin{aligned}
  \textstyle   \min_{x^1,\ldots, x^N} & F(x^1,\ldots, x^N) & \text{s.t.} \  x^i \in \mc{X}(P^i, p^i), \ \forall \ i \in [N],
\end{aligned}
\end{equation}
where the objective  $F:\reals_+^{N\times(T+1)\times S\times A}\mapsto \reals$ is defined as 
\begin{multline}~\label{eqn:game_potential}
   \textstyle  F(x^1, \ldots, x^N)  =\sum_{i,t,s,a} x^i_{tsa}C^i_{tsa} +  \sum_{t,s} k\Big(\sum_{i,a} 2x^i_{tsa} + \prod_{i \in [N]} (1 - \sum_{a}x^i_{tsa}) + \sum_{a} \prod_{i \in [N]}(1 -x^i_{tsa})\Big). 
\end{multline}
\begin{lemma}\label{lem:potential_gradient}
The joint state-action distribution $x$ satisfies~\eqref{eqn:potential_game}'s first order KKT conditions if and only if it corresponds to a Nash equilibria of the MDP congestion game with costs $(\ell^1,\ldots, \ell^N)$~\eqref{eqn:congestion_cost}. 
\end{lemma}
\begin{proof}
From~\cite[Thm.1.3]{calderone2017markov}, the first order KKT conditions of~\eqref{eqn:potential_game} are equivalent to the Nash equilibrium condition if $F$'s gradients satisfy $\textstyle \partial F/\partial x^i_{tsa} = \ell^i_{tsa}(x)$ for all $(i,t,s,a)\in[N]\times[T]\times[S]\times[A]$. We compute $\textstyle \partial F/\partial x^i_{tsa}$ via~\eqref{eqn:game_potential}. With respect to $x^i_{tsa}$, the gradient of $\sum_{i,t,s,a} x^i_{tsa}C^i_{tsa}$ is $C^i_{tsa}$, the gradient of $\textstyle  k\sum_{i,t,s,a} 2x^i_{tsa}$ is $2k$, the gradient of $\textstyle \sum_{t,s} k \prod_{i \in [N]} (1 - \sum_{a}x^i_{tsa})$ is $\textstyle -k\prod_{j\neq i} (1 - \sum_{a}x^j_{tsa})$, and the gradient of $\textstyle  \sum_{t,s,a} k \prod_{i \in [N]} (1 - x^i_{tsa})$ is $\textstyle -k\prod_{j\neq i}(1 - x^j_{tsa})$. Their sum recovers $\ell^i_{tsa}$~\eqref{eqn:congestion_cost} for all $\textstyle (i,t,s,a) \in [N]\times\mc{T}\times[S]\times[A]$. 
\end{proof}
\textbf{Non-convexity}. Distinct from~\cite{li2022congestion,li2019tolling}, congestion costs~\eqref{eqn:congestion_cost} results in a non-convex and multilinear optimization objective~\eqref{eqn:game_potential}. However, the proposed Frank-Wolfe solution ~\cite{li2022congestion, li2019tolling} for finding the Nash equilibria will still converge sublinearly~\cite{lacoste2016convergence}. We refer to~\cite[Alg.1]{li2022congestion} for an algorithm outline.

%% file: flight_model.tex
\section{Stochastic Path Planning for Congested Air Traffic}\label{sec:flight_model}
Air traffic management operates under high operational uncertainty and strict collision risk requirements~\cite{shone2021applications}. Presently, air traffic authorities centrally plan deterministic trajectories and rely on human controllers to resolve local collision risks. 
We use the MDP congestion game model to embed the real-time operation uncertainty into path planning and find global collision risk-free trajectories. 

\textbf{Individual aircraft MDP}. We use an MDP to model deterministic flight plans under operational uncertainty. 
Aircraft $i$'s flight plan is $\textstyle \{(w^i_t, f^i_t) \ | t \in \mc{T}^i \}$, where $w^i_t$ are discrete waypoints used by the European Union Aviation Safety Agency (EASA), $f^i_t$ are discrete flight levels from 0 (sea level) to 450 ($45000$ feet) in increments of 50, and $\mc{T}^i$ are timestamps of the waypoints and flight-levels. 
\begin{figure}[ht!]
    \centering
    \includegraphics[width=0.24\columnwidth]{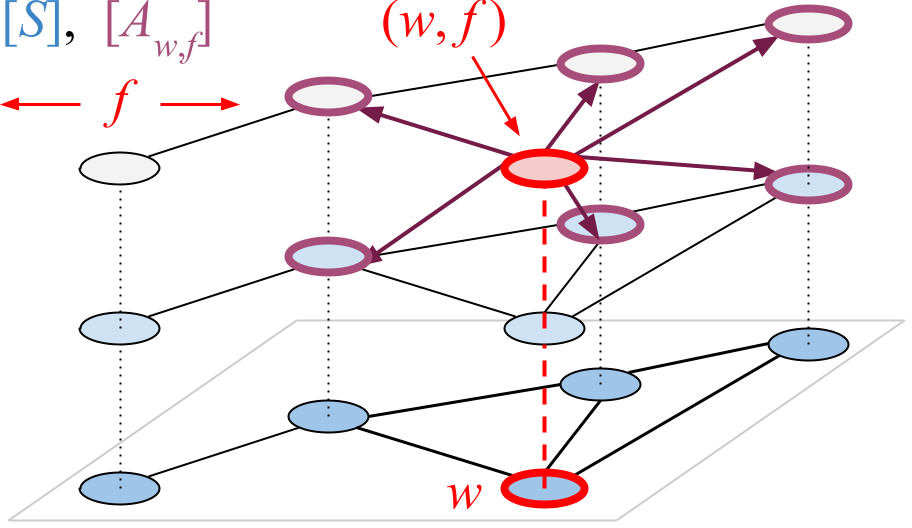}
    \
    \includegraphics[width=0.44\columnwidth]{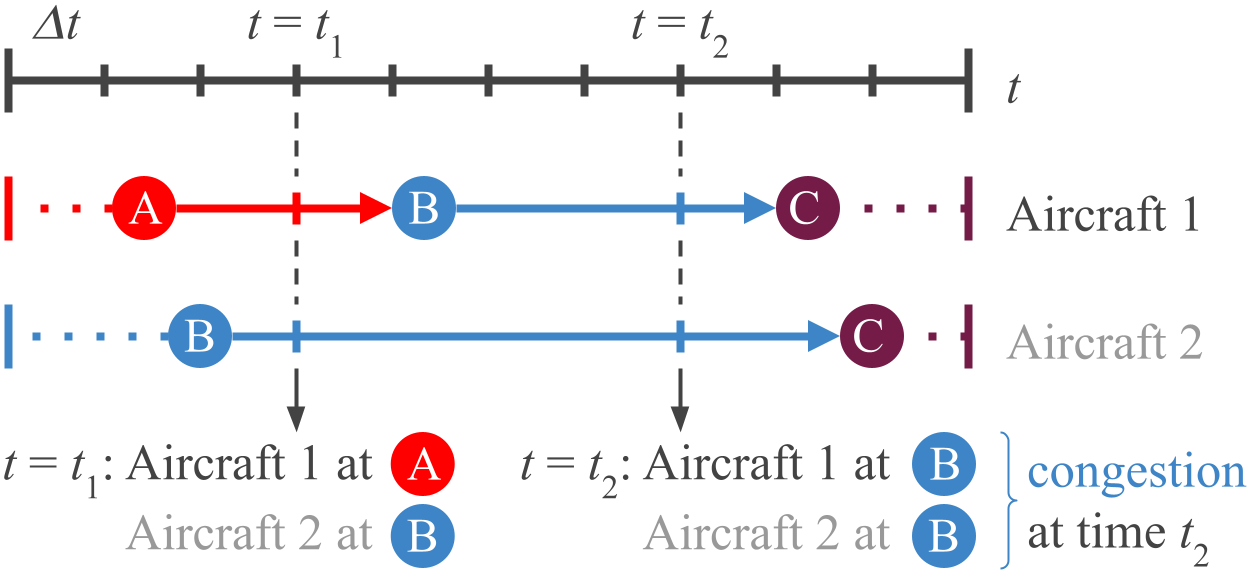}
    \
    \includegraphics[width=0.28\columnwidth]{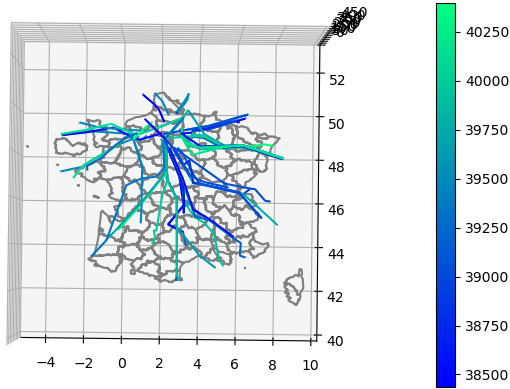}
    \caption{Left: Airspace state-action definitions. Center: Interval-based congestion computation. Right: Expected aircraft trajectories without congestion costs $D^i$~\eqref{eqn:DG_prob}. Colors correspond to time. }
    \label{fig:time_discretization}
\end{figure}

\textbf{Time horizon}. Aircraft $i$'s time horizon is given by $\textstyle \mc{T}^i \cup \{t_L + b\Delta t_{int}\ | 0\leq b \leq B\}$, where $\mc{T}^i$ is from the flight plan, $t_L$ is the planned landing time, and $\Delta t_{int}, B\in \mathbb{N}$ are user-defined parameters.

\textbf{States}. Each state $(w, f) \in [S]$ consists of a waypoint $w$ and a flight level $f$, as shown in Figure \ref{fig:time_discretization}. 

\textbf{Actions}. At state $\textstyle(w,f)$, actions correspond to reaching one of $(w,f)$'s neighbors in the next time step. The set of neighbors is given by $\textstyle \mc{N}(w,f) = \big\{(w', f') \ | \ w' \in \mc{N}(w),  \ f' \in \{f-50, f, f+50\}, \ 0 \leq f'\leq 450 \big\}$, 
    where $\mc{N}(w)$ is the set of reachable waypoints from $w$.  Aircraft cannot loiter at $(w,f)$. The action of going to $(w', f')$ is $\textstyle a_{w',f'}$, such that $\textstyle [A_{w,f}] = \{a_{w',f'} \ | \ (w', f') \in \mc{N}(w,f)\}$.
    
\textbf{Transition Dynamics}. Under action $a_{w', f'}$ from $(w, f)$, an aircraft has $\textstyle \beta$ probability of reaching $(w', f')$ and $\textstyle 1 - \beta$ probability of diverting to another state in $\mc{N}(w,f)$. 
    
 \textbf{Cost}: Each state-action pair $(w, f, a_{w',f'})$ has a flight-dependent \textbf{deviation cost}, given by
    \begin{equation}
       \textstyle   C^i_{t, w, f, a} = d(w^i_t, w) + \alpha_f | f - f^i_t| + L(t, w, f), \ \forall (w,f) \in [S], a \in [A_{w,f}],
    \end{equation}
    where $(w^i_t, f^i_t)$ is aircraft $i$'s planned location at $t$, $d(v, w)\in \mathbb{N}$ is the number of edges between $v$ and $w$, $\alpha_f \in \reals$ is user-defined parameter, and $L:\mc{T}^i\times[S]\mapsto \reals$ is a tardiness cost. 
    If the aircraft plans to land at $(w_T, f_T, T)$, then $\textstyle L(t, w, f) = 0$ if $\textstyle (w,f) = (w_T, f_T)$ or $\textstyle t \leq T$, else $\textstyle L(t, w, f) = c_{tardy}(t-T)$. 
    The expected cost under the flight plan is zero and strictly positive otherwise. Therefore, aircraft are inclined to follow the flight plan in the absence of congestion. 

%% file: atm_results.tex
Based on the individual aircraft MDP model, we build an MDP congestion game for the air traffic plan over France on July $3^{rd}$, $2017$. Between timestamps $39000$ and $41000$, $75$ planes left the Paris airports CDG and ORY to various destinations as shown in Figure~\ref{fig:time_discretization}. The collision risks $D^i$ and $G^i$~\eqref{eqn:DG_prob} can be interpreted as standard aircraft radial/vertical separation and longitudinal separation~\cite{brooker2006longitudinal}, respectively. In our simulations, only $D^i$ increases congestion costs.

\textbf{Interval-based collision risk computation}. Since each aircraft's time stamp is unique, we compute the congestion for time intervals. As shown in Figure~\ref{fig:time_discretization}, aircraft whose time stamp fall into the interval $\textstyle [t_k, t_k + \Delta t_{cong})$ will contribute to the congestion in time interval $k$.  
    
\begin{figure}[ht!]
    \centering
    \includegraphics[width=0.45\columnwidth]{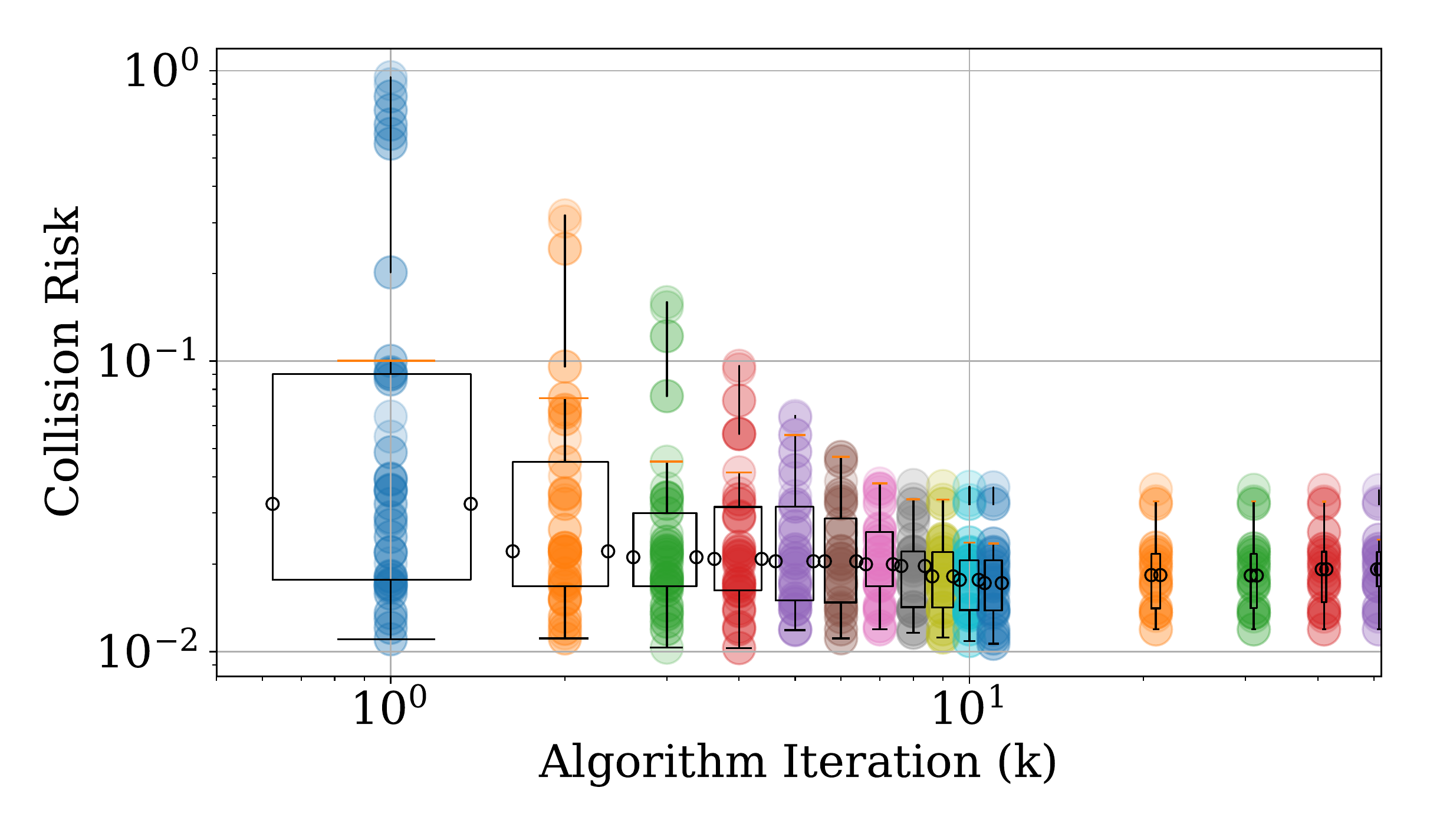}
    \
    \includegraphics[width=0.45\columnwidth]{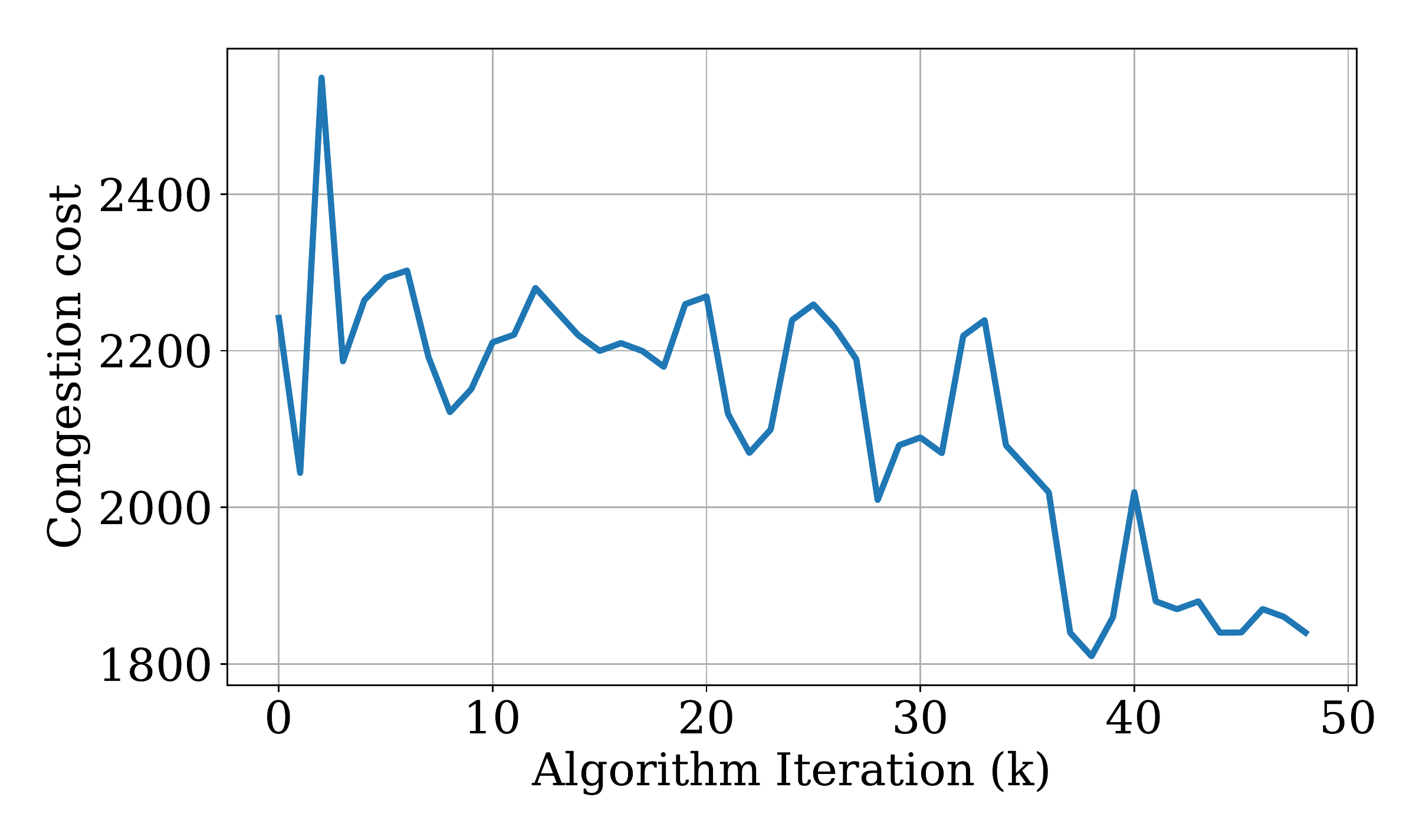}
    \caption{Left: Collision risk as a function of Frank-Wolfe algorithm iteration. Right: Congestion cost $\textstyle \sum_{i,t,s,a}kD_{tsa}^i(x)$~\eqref{eqn:DG_prob} as a function of Frank-Wolfe algorithm iteration. }
    \label{fig:FW_output}
\end{figure}

\textbf{Results and discussion}. We build the individual MDP and the interval-based congestion costs with the following user-defined parameter values: $\Delta t_{int} = 300$, $B = 3$, $\beta = 0.95$, $\alpha_f = 10$, $c_{tardy} = 2$, $\Delta t_{cong} = 19$, and $k = 10$. First, we verify that when solved without congestion cost $D^i$, all individual MDPs result in expected trajectories that match the original flight plan. The results are shown in Figure~\ref{fig:time_discretization}. We then define collision risk as $\sum_{a}x^i_{tsa}D^i_{tsa}(x)$, and found that for multiple flights, the maximum collision risk at any time was greater than $10\%$. The overall spread of collision risks for the original flight plan is shown on the  $x = 10^0$ line in Figure~\ref{fig:FW_output} left. We then augment individual costs with congestion cost $D^i$ and solve for the Nash equilibrium via the Frank Wolfe algorithm from~\cite[Alg.1]{li2022congestion}. The resulting collision risks and objective values are shown in Figure~\ref{fig:FW_output}. In the right figure, we see that the objective value decreases from $2200$ to $1900$ within the first $50$ iterations. Accompanying this, we observe that the maximum collision risks drops from $94\%$ to around $3\%$ within the first $10$ iterations of the Algorithm. Therefore, we conclude that our model was effective in reducing uncertainty-induced collision risks.